\documentclass[11pt]{article}

\RequirePackage[OT1]{fontenc}
\RequirePackage{amsthm,amsmath}
\RequirePackage[numbers]{natbib}

\usepackage[dvips]{graphicx}
\usepackage{verbatim}
\usepackage{graphicx}
\usepackage{amsmath}
\usepackage{amssymb}
\usepackage{mathrsfs}
\usepackage{color}
\usepackage{verbatim}
\usepackage{epstopdf}
\newtheorem{definition}{Definition}
\newtheorem{proposition}{Proposition}

\begin{document}

\title{Sequential Monte Carlo for fractional Stochastic Volatility Models}
\author{Alexandra Chronopoulou \thanks{Research of A.C. supported in part by a start-up fund from
the University of Illinois at Urbana-Champaign and by the Simons Foundation.}\hspace{.2cm}\\
Department of Industrial \& Enterprise Systems Engineering\\ University of Illinois at Urbana-Champaign\\
Konstantinos Spiliopoulos \thanks{Research of K.S. supported in part by a start-up fund from Boston University and by the National Science Foundation (DMS 1550918).}\\
Department of Mathematics and Statistics\\ Boston University \\
}

\maketitle

\bigskip
\begin{abstract}
In this paper we consider a fractional stochastic volatility model, that is a model in which the volatility may exhibit a long-range dependent or a rough/antipersistent behavior. We propose a dynamic sequential Monte Carlo methodology that is applicable to  both long memory and antipersistent processes in order to estimate   the volatility  as well as the unknown parameters of the model. We establish a central limit theorem for the state and parameter filters and we study  asymptotic properties (consistency and asymptotic normality) for the filter. We illustrate our results with a simulation study and we apply our method to estimating the volatility and the parameters of a long-range dependent model for S\& P 500 data.
\end{abstract}


\section{Introduction}

Empirical studies show that the volatility may exhibit correlations of the squared log returns that decay at a hyperbolic rate, instead of an exponential rate as the lag increases, see for example \cite{BCL}, \cite{DGE}. This slow decay cannot  be explained by a GARCH stochastic volatility model or by a stochastic volatility model with jumps. In the literature, this behavior has been described by a class of models that exhibit long-range dependence in the volatility. In contrast to the models where the dependence is introduced in the stock returns, \cite{rogers}, the standard assumption of absence of arbitrage opportunities is preserved.  The first long memory stochastic volatility (LMSV) model was introduced in discrete time by \cite{BCL} and \cite{Ha}. In particular, the authors model the stock returns by:
\begin{equation*}
\begin{cases}
&r_t = \sigma_t\\
&\sigma_t = \sigma \exp\left\{ u_t\right\}, \text{ where }
(1-B)^{d} u_t = \eta_t
\end{cases}
\end{equation*}
where $\sigma>0$, $\epsilon_t$  are independent and identically distributed (iid) $N(0,1)$, $\eta_t$ are iid $N(0, \sigma_{\eta})^{2}$ and independent with $\epsilon_t$ and $d\in(0,1/2)$ is the memory parameter. The feature of long-memory here stems from the fractional difference $(1-B)^d$, where $B$ denotes the lag operator, i.e. $BX_t=X_{t-1}$.

The continuous analogue of the LMSV model was introduced by Comte and Renault in  \cite{comte}  as a continuous-time mean
reverting process in the Hull-White setting. More specifically, their model resembles a classical stochastic volatility model  in which the log returns follow a Geometric Brownian Motion , however, the volatility process is described by a fractional Ornstein-Uhlenbeck process; that is the standard Ornstein-Uhlenbeck process where the Brownian motion is replaced by a fractional Brownian motion.

A large number of recent papers have considered modeling  the volatility in terms of long-range dependent or antipersistent processes. In \cite{comte2}, the authors propose an affine version of the  long memory  model in \cite{comte} for option pricing and they argue that long-range dependence provides an explanation for observations of non-flat term structure in long maturities. In \cite{chrono1}, \cite{chrono2}, the authors propose a pricing tree algorithm in order to compute option prices in the discrete and continuous time frameworks. They also propose a calibration method for determining the memory parameter.  In \cite{bayer}, \cite{gatheral}, the authors propose a rough fractional stochastic volatility model, which is used to explain strong skews or smiles in the implied volatility surface for very short time maturities. In a more recent paper, \cite{garnier}, the authors discuss the case of small volatility fluctuations of both short and long memory models and their impact on the implied volatility and derive a corrected Black-Scholes formula.

In this article, we address two problems: filtering of the unobserved volatility and  parameter estimation  in the case of  a stochastic volatility model that is either rough or antipersistent.
Specifically, we adapt a Sequential Monte Carlo algorithm in the non-Markovian framework and then we estimate the parameters online, following an idea initially introduced by  \cite{liuwest}.

Sequential Monte Carlo (SMC) methods, also known as particle filters or recursive Monte Carlo filters,  arose from the seminal work of Gordon, Salmond and Smith (1993), \cite{gordon}. SMC methods are iterative algorithms that are based on the dynamic evolution and update of discrete sets of sampled state vectors, that are referred to as particles, and are associated with properly selected weights.  
For an introduction to the SMC literature, there are several  reviews and tutorials, including, the edited volume by Doucet, Freitas and Gordon (2001), \cite{doucet1} and Kantas et al. (2010), \cite{kantas}.

In order to keep the adaptive nature of  the SMC algorithm, we incorporate the parameter estimation procedure  in the algorithm by treating the parameters as artificially random, sampled from a kernel density. Generally speaking, the approach relies on augmenting the unobserved state by considering the parameter as an unobserved state.
 We do not assume that the unknown parameter is fully random, and thus we do not have an additional MCMC step in the algorithm. We emphasize here that this is an important point, due to the heavy computational overhead caused by the long-memory issue.  Also, as it is well known, using SMC for maximum likelihood based methods for estimating parameters like the Hurst parameter is difficult.

 The output of the algorithm is  a provably asymptotically  consistent estimator along with the corresponding variance.  One of the appeals of the proposed method is the simplicity in its application, which becomes in particularly important due to the memory issue. However, we do mention here that overdispersion due to the artificial random evolution of the parameter is an issue that comes up with such methods, see for example \cite{chopinsmc2}. But, as we shall see, the overdispersion that the artificial random evolution of the parameter introduces, can be minimized by properly tuning the parameters of the mixture distribution used for sampling the parameter.  Related works using different methods can be found in \cite{BeskosDureauKalogeropoulos2013,LysyPillai2013}.

The study of the asymptotic properties of SMC methods and the asymptotic properties of the filter  as the number of particle increases is a quite hard problem. A form of consistency of the filter, when the number of particles tends to infinity is a common result in the majority of the literature, while central limit theorem type of results are fewer. One can refer, for example, to the work of Del Moral and Guionnet (1999), \cite{delmoral1}.
 To the best of our knowledge, the most general results in the literature can be found in Chopin (2004), \cite{chopin}, Douc and Moulines (2008), \cite{doucMoulines}, and K\"{u}nsch (2005), \cite{kunsch}. In these works and under slightly different conditions, law of large numbers and central limit theorems are derived that apply to most sequential Monte Carlo techniques in the literature.

The rest of the paper is organized as follows: In Section 2, we introduce the mathematical formulation of the problem. In Section 3,  we introduce the SISR (Sequential Importance Sampling with Resampling) method with parameter learning and we present the theoretical results for the proposed parameter estimators. In Section 4, we study the performance of both methods using simulated data. In Section 5, we  apply our method in estimating the unobserved volatility of a discrete-time stochastic volatility model with long-range dependence for S\& P 500 data. Finally, we summarize our results in  Section 6.

\section{Mathematical Framework}

Consider a state-space model in which the state vector  is denoted by $\{X_t\}_{t\geq 1}$ and the observations  $\{Y_t\}_{t \geq 1}$ are obtained sequentially in time. In addition, we assume that the state vector depends on an unknown, but fixed, parameter vector that we denote by $\theta$. In the sequel, we use the notation $X$ or $Y$ for the random variables and  $x$ or $y$ for  the corresponding realized values.

Unlike other models in the literature, we do not assume that the state vector  $\{X_t\}_{t\geq 1}$ is a Markovian process. Instead, we consider the case in which the unobserved process is not necessarily Markovian, with particular interest in the  long-range dependent case. Formally,  long-range dependence or long-memory is  defined as follows:
\begin{definition}
For a stationary process $\{X_t\}_{t\geq 1}$, there exists a parameter (Hurst index) $H\in(\frac{1}{2},1)$, such that
\begin{equation}\label{eq:LRD}
\lim_{t\rightarrow \infty} \frac{Corr(X_t, X_1)}{ct^{2-2H}} = 1,
\end{equation}
where $\rho(h) := Corr(X_{t}, X_{t-h})$ is the autocorrelation function of the process.
\end{definition}
 When $H=\frac{1}{2}$, then the process is Markovian, so this is a generalization of the models that are treated in the SMC literature. Equivalently, long-range dependence  implies that the autocorrelation function $\rho(h)$ of a long-range dependent process is non-summable, that is  $\sum_{h=1}^{\infty} \rho(h) = \infty$. If the auto-correlation function is summable, then the process has what is called antipersistence, in which case $H\in(0, \frac{1}{2})$.

Formally, at time $t$, the state-space model is specified by the \textit{observation equation} that is determined by the observation density
\[ p(y_t|x_t; \theta)\]
and the \textit{state equation} given by the conditional density
\[p(x_t| x_{t-1},\ldots, x_1;\theta),\]
where $\theta \in \Theta$ is an unknown vector of parameters, and $\Theta \subset \mathbb{R}^{d}$  is  open.  We assume that  the observations $\{Y_t\}_{t\geq 1}$ are conditionally independent given $\{X_t\}_{t\geq 1}$ and  that the long-range dependent process $\{X_t\}_{t\geq 1}$ has known initial density  $X_1 \sim \mu (x; \theta)$.

In this article, our goal is to use simulation for online filtering. In other words, we want to learn about the current state $X_t$ given  available information up to time $t$, which reduces to estimating the probability distribution function
\[p(x_t|y_{1},\ldots,y_{t}; \theta), \;\; t=1, \ldots,n,\]
where $\{y_t\}_{t\geq 1}$ are the observations up to time $t$. However, since we assume that the parameter $\theta$ is unknown, at the same time we also want to estimate $\theta$.

\section{Sequential Monte Carlo Filtering}

As we mentioned above,  apart from filtering for the unobserved states we also want to estimate the unknown parameter vector $\theta$ on which the state vector depends. Our approach will be to consider $\theta$ as an additional state and thus our goal will be to estimate the posterior distribution  $\{p(x_{1:t}, \theta|y_{1:t})\}_{t\geq 1}$ given by
\begin{align}\label{eq:postD_theta}
p(x_{1:t}; \theta|y_{1:t}) \;\;&\propto \;\;p(x_{1:t}, y_{1:t}, \theta) \nonumber \\
&\propto \;\; p(x_{1})\cdot  p(x_2|x_1; \theta) \cdot \ldots \cdot p(x_{n}|x_{n-1}, \ldots, x_1; \theta) \cdot \prod_{i=1}^{t} p(y_{i}|x_{i}; \theta) \cdot p(\theta),
\end{align}
where $p(\theta|y_{t})$ is a prior density for the parameter vector $\theta$, see Subsection \ref{SS:OnLineEstimation}. If the parameter is known, then the density is degenerate. Therefore, the additional difficulty here is that we need to compute or approximate the theoretical density function $p(\theta|y_t)$.

\subsection{On-line Parameter Estimation}\label{SS:OnLineEstimation}

One approach in the literature (\cite{flury}, \cite{gordon}, \cite{west1}, \cite{liuwest}), is to consider that $\theta$ is not fixed and assume that it artificially evolves in time, for example
\[\theta_t = \theta_{t-1} + e_t\]
where $e_t$ is an artificial white noise with decreasing variance. Then, at each time $t$,  $p(\theta|y_t)$ will be updated inside the SISR algorithm in order to incorporate the additional information that is obtained.

As it was discussed in \cite{liuwest}, this approach leads to artificial variance inflation, since the parameter is not truly random.  However, the  use of  a kernel density estimate with shrinkage correction can control this artificial over-dispersion.

More specifically, standing at time $t$,  we  approximate $p(\theta|y_{t})$ by a set of samples $\theta_{(t)}^{(j)}$ and weights $\omega_{t}^{(j)}$ using a discrete Monte Carlo. The index $t$ in $\theta$ is in parenthesis to
indicate that $\theta$ does not evolve in time, but that its value  is drawn from the posterior density $p(\theta|y_t)$ at time $t$.
Then, the smooth kernel density with shrinkage correction will be of the form
\[p(\theta|y_t) \approx \sum_{j=1}^{N} w_{t}^{(j)} \mathcal{N}\left(\theta|m_{t}^{(j)}, h^2 V_t\right),\]
where $\mathcal{N}(\cdot|m,S)$ denotes a multivariate Normal density with mean $m$ and variance $S$. So, essentially,  $p(\theta|y_t) $ is approximated by a mixture of normals with  mean $m_t^{(j)}$ and variance $h^2 V_t$, weighted by sample weights  $\omega_t^{(j)}$.  The kernel location is specified by
\[
m_t^{(j)} = \alpha \; \theta_{t}^{(j)} + (1-\alpha) \bar{\theta}_{t},
\]
where $\alpha = \sqrt{1-h^2}$ and $\bar{\theta}_{t}$ denotes the average over all parameter samples (essentially a sum over $i$). Regarding $h$, a typical choice would be a decreasing  function of the sample size, but if one wants to control the loss of information then  $h^2 = q - ((3\delta - 1)/2\delta)^2$, where $\delta$ is a discount factor typically around $0.95-0.99$ and $\alpha$  becomes $\alpha = (3\delta-1)/2\delta$.

Therefore, the SISR algorithm is adjusted in order to incorporate the update of $\theta$. The key idea of our approach that also allows us to establish asymptotic consistency and normality of the estimators, is to re-formulate the weights so that they represent the joint posterior $p\left(x_t^{(j)}, \theta_{(t)}^{(i)}|y_{1:t}\right)$  and update $\theta$ along with the state vector $X$.

Before presenting the algorithm, let us define the un-normalized weight functions. For $i=1, \ldots, N$, and $t=1$ set
\begin{align}
w_1\left( \tilde{X}_{1,1}^{(i)};\; \theta_{(1)}^{(i)}\right) &=
\frac{p\left(X_{1,1}^{(i)}, \theta_{(1)}^{(i)}\bigl| \;y_{1}\right)}{q_{1}\left(\tilde{X}_{1,1}^{(i)}\right)p_{1}\left (\theta_{(1)}^{(i)} \right)}  \propto  \frac{p_{1} \left (\theta_{(1)}^{(i)} \right) \mu\left(X_{1,1}^{(i)}| \theta_{(1)}^{(i)}\right)  \;p\left(y_{1} \bigl|\tilde{X}_{1,1}^{(i)}, \theta_{(1)}^{(i)}\right) }
{q_{1}\left(\tilde{X}_{1,1}^{(i)}\right)p_{1}\left (\theta_{(1)}^{(i)} \right)}\nonumber\\
&=  \frac{ \mu\left(X_{1,1}^{(i)}| \theta_{(1)}^{(i)}\right)  \;p\left(y_{1} \bigl|\tilde{X}_{1,1}^{(i)}, \theta_{(1)}^{(i)}\right) }
{q_{1}\left(\tilde{X}_{1,1}^{(i)}\right)}.\label{Eq:WeightFcnT1}
\end{align}
Generally for $t\geq 1$
{\small
\begin{align}
w_t^{(i)}&=w_t\left(X^{(i)}_{1:t-1,t-1},  \tilde{X}_{t,t}^{(i)}; \;\theta_{(t)}^{(i)} \right)\nonumber\\
 &= \frac{p\left( X^{(i)}_{1:t-1,t-1}, \tilde{X}_{t,t}^{(i)}, \theta_{(t)}^{(i)}  \bigl| y_{1:t} \right)} {p\left(X^{(i)}_{1:t-1,t-1} | y_{1:t-1} ; \theta_{(t)}^{(i)} \right) \; q_t\left(  \tilde{X}_{t,t}^{(i)}, \theta_{(t)}^{(i)}| X^{(i)}_{1:t-1,t-1}\right)} \nonumber\\
 &  = \frac{p\left(y_t|X^{(i)}_{1:t-1,t-1},\tilde{X}_{t,t}^{(i)},   \theta_{(t)}^{(i)} \right ) \; p\left(\tilde{X}_{t,t}^{(i)} |X^{(i)}_{1:t-1,t-1},  y_{1:t-1}, \theta_{(t)}^{(i)} \right) p\left(X^{(i)}_{1:t-1,t-1} | y_{1:t-1},  \theta_{(t)}^{(i)} \right) p\left(  \theta_{(t)}^{(i)}|y_{1:t-1},\right) } { p\left(X^{(i)}_{1:t-1,t-1} | y_{1:t-1} ; \theta_{(t)}^{(i)} \right) q_t\left( \tilde{X}_{t,t}^{(i)}\bigl | X^{(i)}_{1:t-1,t-1},  \theta_{(t)}^{(i)}\right)p\left(  \theta_{(t)}^{(i)}\right)}\nonumber\\
 &  = \frac{p\left(y_t|X^{(i)}_{1:t-1,t-1},\tilde{X}_{t,t}^{(i)},   \theta_{(t)}^{(i)} \right ) \; p\left(\tilde{X}_{t,t}^{(i)} |X^{(i)}_{1:t-1,t-1},  y_{1:t-1}, \theta_{(t)}^{(i)} \right) p\left(X^{(i)}_{1:t-1,t-1} | y_{1:t-1},  \theta_{(t)}^{(i)} \right)  } { p\left(X^{(i)}_{1:t-1,t-1} | y_{1:t-1} ; \theta_{(t)}^{(i)} \right) q_t\left( \tilde{X}_{t,t}^{(i)}\bigl | X^{(i)}_{1:t-1,t-1},  \theta_{(t)}^{(i)}\right)}.\label{Eq:WeightFcnT}
\end{align}
}

Then, the algorithm is given by:

\noindent \textbf{ At time $t=1$}
\begin{enumerate}
\item[(a)] \textit{Sampling}\\
For $i=1, \ldots, N$, sample $\tilde{X}_{1,1}^{(i)}    \sim \,\, q_{1}(\cdot)$ and $\theta_{(1)}^{(i)} \sim p_{1}(\cdot)$.

\item[(b)] \textit{Re-Sampling }\\
For $i=1, \ldots, N$, let $w_1\left( \tilde{X}_{1,1}^{(i)};\; \theta_{(1)}^{(i)}\right)$ be defined by (\ref{Eq:WeightFcnT1}),
 normalize $W_{1}^{(i)} = \frac{w_1^{(i)}}{\sum_{i=1}^{N} w_1^{(i)}}$, such that $\sum_{i=1}^{N} W_{1}^{(i)}=1$ and re-sample
\[X_{1,1}^{(i)} \sim \sum_{j=1}^{N} W_{1}^{(j)} \delta_{\tilde{X}_{(1,1)}^{(j)}}\left( dx_1\right).\]
\end{enumerate}

\noindent \textbf{ At time $t, t\geq 2$ (step $t-1 \rightarrow t$)}
\begin{enumerate}
\item[(a)] \textit{Sampling} \\
For $i=1, \ldots, N$,  set
\begin{align*}
\tilde{X}^{(i)}_{t,1:t-1} &= X^{(i)}_{t-1,1:t-1}\\
m_{t-1} &= \alpha \theta_{(t-1)}^{(i)} + (1-\alpha) \bar{\theta}_{(t-1)}
\end{align*}
where $\bar{\theta}_{(t-1)} =  \sum_{i=1}^{N} W_{t-1}^{(i)} \theta_{(t-1)}^{(i)}$,  sample
\[\theta_{(t)}^{(i)} \sim \mathcal{N}(m_{t-1}^{(i)} | h^{2} V_{t-1}), \]
and
\[
\tilde{X}_{t,t}^{(i)} \sim q_{t} \left( \cdot \bigl| X^{(i)}_{1:t-1,t-1}; \theta_{(t)}^{(i)}\right)
\]
where $V_{t-1} = \frac{1}{N-1} \sum_{i=1}^{N}\left( W_{t-1}^{(i)} \theta_{(t-1)}^{(i)} - \bar{\theta}_{(t-1)} \right)^2$ .\\

\item[(b)] \textit{Re-Sampling}\\
For $i=1, \ldots, N$, let $w_t^{(i)}=w_t\left(X^{(i)}_{1:t-1,t-1},  \tilde{X}_{t,t}^{(i)}; \;\theta_{(t)}^{(i)} \right)$ be defined by (\ref{Eq:WeightFcnT}),
and normalize $W_{t}^{(i)} = \frac{w_t^{(i)}}{\sum_{i=1}^{N} w_t^{(i)}}$, such that $\sum_{i=1}^{N} W_{t}^{(i)}=1$.\\

For $i=1, \ldots, N$, re-sample
\[X_{1:t,t} \sim \pi^{N}_{t}(dx_{1:t}), \quad\text{where}\quad\pi^{N}_{t}(dx_{1:t})=\sum_{j=1}^{N} W_{t}^{(j)} \delta_{X^{(j)}_{1:t-1,t-1}, \tilde{X}_{t,t}^{(j)}} (dx_{1:t}).
\]
and set
\[
 \bar{\theta}_{(t)} =  \sum_{i=1}^{N} W_{t}^{(i)} \theta_{(t)}^{(i)}
\]
\end{enumerate}

\textbf{Output}
The filtering distribution $p(dx_{1:t}|y_{1:t})$ is approximated by
\[
\pi^{N}(dx_{1:t})=\sum_{j=1}^{N} W_{t}^{(j)} \delta_{X^{(j)}_{1:t-1,t-1}, \tilde{X}_{t,t}^{(j)}} (dx_{1:t}), \quad \text{or} \quad
\tilde{\pi}^{N}(dx_{1:t})=\frac{1}{N}\sum_{j=1}^{N} \delta_{X^{(j)}_{1:t,t}}(dx_{1:t}).
\]
and the estimator for $\theta$ is $\bar{\theta}_{(t)}$. We also record the approximation for the combined distribution
$\pi^{\theta}(dx_{1:t}, d\theta_{(t)})=p(dx_{1:t}, d\theta_{(t)}|y_{1:t})$ which is approximated by
\[
\pi^{N,\theta}(dx_{1:t}, d\theta_{(t)})=\sum_{j=1}^{N} W_{t}^{(j)} \delta_{X^{(j)}_{1:t-1,t-1}, \tilde{X}_{t,t}^{(j)}, \theta_{(t)}^{(j)}} (dx_{1:t}d\theta_{(t)}).
\]

\subsection{Convergence Results}

Let us now study the convergence properties of this algorithm.
Let $\phi:\mathcal{X}\times\Theta\mapsto\mathbb{R}$ be an appropriate test function. Notice now that the SISR algorithm provides us with the estimator
\[
\hat{\phi}^{N}_{t}=\int\phi(x_{1:t}, \theta_{(t)})\pi^{N,\theta}(dx_{1:t}d\theta_{(t)})=\sum_{i=1}^{N}W_{t}^{i}\phi\left(X^{(i)}_{1:t-1,t-1}, \tilde{X}_{t,t}^{(i)},\theta_{(t)}^{(i)}\right).
\]
It is relatively straightforward to see that  $\hat{\phi}^{N}_{t}$ is estimating
\[
\bar{\phi}_{t}=\int\phi(x_{1:t},\theta_{(t)})p(x_{1:t}, \theta_{(t)}|Y_{1:t})dx_{1:t}d\theta_{(t)}
\]
where $\theta_{(t)}$ is the value that is drawn from the posterior density $p(\theta|y_{t})$ at time $t$.

So, it is natural to quantify the performance of the algorithm by studying the convergence of $\hat{\phi}^{N}_{t}$ to $\bar{\phi}_{t}$ as $N\rightarrow\infty$.
We define the set of appropriate test functions under which a central limit theorem can be established, appropriately formulated for our case of interest:
\begin{align}
 \Phi_{t}&=\left\{\phi:\mathcal{X}\mapsto\mathbb{R} \text{ measurable}:\text{ there exists }\delta>0 \text{ such that }
 \mathbb{E}_{\pi_{t}}\left\Vert W_{t}\phi\right\Vert^{2+\delta}<\infty,\right.\nonumber\\
 &\hspace{1cm}\left.\text{ and }x_{1:(t-1)}\mapsto \mathbb{E}_{\bar{\rho}(x_{1:(t-1),\cdot})}(W_{t}\phi)^{2+\delta}\text{ is in }\Phi_{t-1} \right\}.\label{Eq:AdmissibleTestFunctions}
\end{align}
Following the proof of  the central limit theorem results of  \cite{chopin,johansen1} for the Markovian case, without the parameter estimation aspect,
the following result is derived.

\begin{proposition}\label{P:CLT_longRange_WithParameter}
Let us assume that there exists $\delta>0$ such that for every $t<\infty$  $\mathbb{E}_{\pi^{\theta}_{t}}\left\Vert W_{t}\right\Vert^{2+\delta}<\infty$ and consider
$\phi\in\Phi_{t}$.
Then, we get
\begin{equation*}
\sqrt{N}\left(\hat{\phi}^{N}_{t}-\bar{\phi}_{t}\right)\Rightarrow \mathcal{N}\left(0,\sigma^{2}_{t}(\phi)\right)
\end{equation*}
as $N\rightarrow\infty$, where at time $t=1$
\begin{align*}
\sigma^{2}_{1}(\phi)&=\int\frac{p^{2}(x_{1},\theta_{(1)}|y_{1})}{q_{1}(x_{1})p_{1}\left(\theta_{(1)}\right)}\left(\phi(x_{1},\theta_{(1)})-\bar{\phi}_{1}\right)^{2}dx_{1}d\theta_{(1)}
\end{align*}
and for $t>1$
\begin{align*}
&\sigma^{2}_{t}(\phi)=\int\frac{p^{2}(x_{1},\theta_{(1)}|y_{1:t})}{q_{1}(x_{1})p_{1}\left(\theta_{(1)}\right)} \left(\int \phi(x_{1:t},\theta_{(t)})p(x_{2:t},\theta_{(t)}|y_{2:t},x_{1})dx_{2:t}d\theta_{(t)}-\bar{\phi}_{t}\right)^{2}dx_{1}d\theta_{(1)}\nonumber\\
&+\sum_{k=2}^{t-1}\int \frac{p^{2}(x_{1:k},\theta_{(k)}|y_{1:t})}{p(x_{1:(k-1)}|y_{1:(k-1)};\theta_{(k)}) q_{k}(x_{k},\theta_{(k)}|x_{1:(k-1)})}\times\nonumber\\
&\qquad\times\left(\int \phi(x_{1:t},\theta_{(t)})p(x_{(k+1):t},\theta_{(t)}|y_{(k+1):t},x_{1:k})dx_{(k+1):t}d\theta_{(t)}-\bar{\phi}_{t}\right)^{2}dx_{1:k}d\theta_{(k)}\nonumber\\
&+\int \frac{p^{2}(x_{1:t},\theta_{(t)}|y_{1:t})}{p(x_{1:(t-1)}|y_{1:(t-1)};\theta_{(t)}) q_{t}(x_{t},\theta_{(t)}|x_{1:(t-1)})}\left( \phi(x_{1:t},\theta_{(t)})-\bar{\phi}_{t}\right)^{2}dx_{1:t}d\theta_{(t)}.\nonumber
\end{align*}
\end{proposition}
\begin{proof}
The proof follows from Proposition A.1.1 of \cite{johansen1} after making the adequate identifications. Indeed, for a general sequential importance sampling algorithm with weights
$W_{t}$, Proposition A.1.1 of \cite{johansen1} implies that the formula for the variance in question is given by
\[
 V_{t}(\phi)=\sum_{k=1}^{t-1}\text{Var}_{\rho_{k}}\left[W_{k}\left(\mathbb{E}_{\pi^{\theta}_{t}}\left[\phi_{t}|X_{1:k}\right]-\bar{\phi}_{t}\right)\right]
 +\text{Var}_{\rho_{t}}\left[W_{t}(\phi(x_{1:t},\theta_{(t)})-\bar{\phi}_{t})\right].
\]

In our case we have
\begin{align}
\pi^{\theta}_{t}(dx_{1:t},d\theta_{(t)})&=p(dx_{1:t},d\theta_{(t)}|y_{1:t})\nonumber\\
\rho_{t}(dx_{1:t},d\theta_{(t)})&=p(dx_{1:(t-1)}|y_{1:(t-1)})q_{t}(dx_{t},d\theta_{(t)}|x_{1:(t-1)})  \nonumber
\end{align}
and the weights take the form
\begin{align*}
W_{t}(x_{1:t},\theta_{(t)})=\frac{d\pi^{\theta}_{t}}{d\rho_{t}}(x_{1:t},\theta_{(t)}) =\frac{p(x_{1:t},\theta_{(t)}|y_{1:t})}{p(x_{1:(t-1)}|y_{1:(t-1)})q_{t}(x_{t},\theta_{(t)}|x_{1:(t-1)})}
\end{align*}
Plugging these expressions in the formula for $V_{t}(\phi_{t})$ one immediately recovers the form of $\sigma^{2}(\phi_{t})$, completing the proof of the proposition.
\end{proof}

\subsection{Statistical properties of the parameter estimator}
Essentially, $\theta$ is viewed as an augmented state variable. Proposition \ref{P:CLT_longRange_WithParameter} quantifies the convergence of the filter,
but it does not discuss the
statistical properties of  $\bar{\theta}^{N}_{(t)}$. Let us recall that
\[
\bar{\theta}^{N}_{(t)} =  \sum_{i=1}^{N} W_{t}^{(i)} \theta_{(t)}^{(N,i)}
\]
where
\begin{align*}
\theta_{(t)}^{(N,i)} &\sim \mathcal{N}(m_{t-1}^{(N,i)} | h^{2} V_{t-1}^{N}), \nonumber \\
m_{t-1}^{(N,i)} &= \alpha \theta_{(t-1)}^{(N,i)} + (1-\alpha) \bar{\theta}_{(t-1)}^{N}\nonumber\\
V_{t-1}^{N}& = \frac{1}{N-1} \sum_{i=1}^{N}\left( W_{t-1}^{(i)} \theta_{(t-1)}^{(N,i)} - \bar{\theta}_{(t-1)}^{N} \right)^2
\end{align*}

By inspecting the algorithm it becomes clear that the convergence properties of $\bar{\theta}^{N}_{(t)}$ as $N\rightarrow\infty$ is described by a statement very
similar to that
of Proposition \ref{P:CLT_longRange_WithParameter} after making the choice $\phi(x_{1:t},\theta_{(t)})=\theta_{(t)}$. Proposition \ref{P:CLT_Parameter} shows that at time $t$, the estimator for $\theta$, $\bar{\theta}^{(N)}_{(t)}$ converges to $\bar{\theta}_{(t)}=\int\theta p(\theta_{(t)}|y_{1:t})d\theta$ as $N\rightarrow\infty$.
\begin{proposition}\label{P:CLT_Parameter}
Let us assume that there exists $\delta>0$ such that for every $t<\infty$  $\mathbb{E}_{\pi^{\theta}_{t}}\left\Vert W_{t}\right\Vert^{2+\delta}<\infty$ and consider the function identity
$\theta_{(t)}\mapsto \theta_{(t)}$, assuming that it belongs to the set of appropriate test functions $\Phi_{t}$ defined in (\ref{Eq:AdmissibleTestFunctions}).
Let us also define the mean of the posterior distribution $p(\theta_{(t)}|y_{1:t})$
\begin{equation*}
 \bar{\theta}_{(t)}=\int\theta p(\theta_{(t)}|Y_{1:t})d\theta.
\end{equation*}
Then, we get
\begin{equation*}
\sqrt{N}\left(\bar{\theta}^{(N)}_{(t)}-\bar{\theta}_{(t)}\right)\Rightarrow \mathcal{N}\left(0,\sigma^{2}_{t}(\theta)\right)
\end{equation*}
as $N\rightarrow\infty$. 
 The asymptotic variance $\sigma^{2}(\theta_{(t)})$ is defined as follows. At time $t=1$
\begin{align*}
\sigma^{2}_{1}(\theta)&=\int\frac{p^{2}(x_{1},\theta_{(1)}|y_{1})}{q_{1}(x_{1})p_{1}\left(\theta_{(1)}\right)}\left(\theta_{(1)}-\bar{\theta}_{(1)}\right)^{2}dx_{1}d\theta_{(1)}
\end{align*}
and for $t>1$
\begin{align}
&\sigma^{2}_{t}(\theta)=\int\frac{p^{2}(x_{1},\theta_{(1)}|y_{1:t})}{q_{1}(x_{1})p_{1}\left(\theta_{(1)}\right)} \left(\int \theta_{(t)}p(x_{2:t},\theta_{(t)}|y_{2:t},x_{1})dx_{2:t}d\theta_{(t)}-\bar{\theta}_{(t)}\right)^{2}dx_{1}d\theta_{(1)}\nonumber\\
&+\sum_{k=2}^{t-1}\int \frac{p^{2}(x_{1:k},\theta_{(k)}|y_{1:t})}{p(x_{1:(k-1)}|y_{1:(k-1)},\theta_{(k)}) q_{k}(x_{k},\theta_{(k)}|x_{1:(k-1)})}\times\nonumber\\
&\qquad\times\left(\int \theta_{(t)}p(x_{(k+1):t},\theta_{(t)}|y_{(k+1):t},x_{1:k})dx_{(k+1):t}d\theta_{(t)}-\bar{\theta}_{(t)}\right)^{2}dx_{1:k}d\theta_{(k)}\nonumber\\
&+\int \frac{p^{2}(x_{1:t},\theta_{(t)}|y_{1:t})}{p(x_{1:(t-1)}|y_{1:(t-1)},\theta_{(t)}) q_{t}(x_{t},\theta_{(t)}|x_{1:(t-1)})}\left( \theta_{(t)}(x_{1:t})-\bar{\theta}_{(t)}\right)^{2}dx_{1:t}d\theta_{(t)}.\label{Eq:VarianceEstimator}
\end{align}
\end{proposition}
\begin{proof}
This proposition is essentially a special case of Proposition \ref{P:CLT_longRange_WithParameter} with $\phi(x_{1:t},\theta_{(t)})=\theta_{(t)}$. We notice that
 \begin{align}
  \bar{\theta}_{(t)}&=\int\theta_{(t)}p(x_{1:t}, \theta_{(t)}|y_{1:t})dx_{1:t}d\theta_{(t)}\nonumber\\
  &=\int\theta_{(t)}p(x_{1:t}|\theta_{(t)}, Y_{1:t})p(\theta_{(t)}|y_{1:t})dx_{1:t}d\theta_{(t)}\nonumber\\
  &=\int\theta p(\theta_{(t)}|y_{1:t})d\theta\nonumber
 \end{align}
which is the mean of the posterior distribution $p(\theta_{(t)}|y_{1:t})$, as claimed. 
\end{proof}
 The next natural question to ask is whether $\bar{\theta}_{(t)}$ is a consistent estimator of $\theta$ as $t\rightarrow\infty$. Let us recall from Subsection \ref{SS:OnLineEstimation} that choosing $\delta$ around $0.95-0.99$ implies for the tuning parameters $(\alpha,h)\approx(1,0)$. Subsequently, this means, for every fixed $t>0$, that the mean $m_t^{(N,i)}\approx\theta_{t}^{(N,i)}$ and the variance $h^2 V_t\approx 0$, which then implies that for every $t>0$, $\theta_{t}^{(N,i)}\approx \theta_{t-1}^{(N,i)}$.  Notice that if the tuning parameters $(\alpha,h)=(1,0)$ then the distribution $p(\theta_{(t)}|y_{1:t})$ coincides with the posterior distribution $p(\theta|y_{1:t})$, as the parameter does not involve artificially in time any more. Hence, in this case,
by Doob's consistency theorem, see for example Theorem 10.10 in \cite{VanDerVaart}, we would have that  the sequence of posterior measures $\mathbb{P}_{\bar{\theta}|y_{1:t}}$ is consistent under $\theta$ if the model is identifiable, i.e., if
$\mathbb{P}_{\theta_{1}}\neq \mathbb{P}_{\theta_{2}}$ for $\theta_{1}\neq\theta_{2}$. In other words, in such a case,
for every prior probability measure $\Pi$ on $\Theta$ the sequence of posterior measures
$\mathbb{P}_{\bar{\theta}|y_{1:t}}$ is  consistent for $\Pi-$almost every $\theta$. However, in the algorithm and for the purposes of dealing with the issue of degeneracy of the particles, we artificially evolve the unknown parameter which means that we take the tuning parameters to be $(\alpha,h)\approx (1,0)$ but not exactly equal to $(1,0)$. As we shall see from the simulation studies of Section \ref{S:Simulations}, this is sufficient to guarantee that the parameter is being consistently estimated, as expected.

\section{Simulation Results}\label{S:Simulations}

\subsection{Fractional ARIMA process}

The \emph{fractional ARIMA (AutoRegressive Integrated Moving Average)} process was proposed by Box and Jenkins, \cite{box}, in 1970 and has been very popular in applied time series. A fractional ARIMA($p, d, q$) process is formally defined as follows (due to Granger and Joyeux, \cite{granger}):
\begin{definition}
Let $\varphi(\cdot)$ and $\vartheta(\cdot)$ be polynomials of orders $p$ and $q$
respectively and $X_{t}$  a stationary process such that
\[
\varphi (B)(1-B)^{d}X_{t}=\vartheta(B) \eta_{t}.
\]
$d\in(-1/2,1/2)$ and and  $(\eta_{t})_{t \geq 0}$ is a sequence of iid variables with mean $0$ and variance $1$. Then, the process $\{X_{t}\}_{t\geq 0}$ is called a fractional
ARIMA($p,d,q$) process.
\end{definition}
In contrast to the classical ARIMA($p, d, q$) process, where  the parameter $d$ is an integer, in the fractional case $d$ is a real valued parameter with values between $(-1/2, 1/2)$. It is called the fractional integration parameter and is related to the Hurst index, $H$ in \eqref{eq:LRD}, via $d = H - \frac{1}{2}$. $B$ denotes the lag or backshift operator, and
 \[(1-B)^d = \sum_{k=0}^{\infty} \binom{d}{k} (-1)^k B^k,\]
where the sum is taken over an infinite number of indices. The fractional ARIMA process is long-range dependent when $d>0$, while the upper bound on $d$ is needed to ensure that the process is stationary.  More details regarding these models can be found in Beran \cite{beran}.

In our framework, we consider a state-space model in which the unobserved process is modeled by a Fractional ARIMA$(p, d, q)$ process. Specifically, the state-space model is defined as follows
\begin{equation}\label{eq:lmsv}
\begin{cases}
& Y_{t} = \sigma\left(\frac{X_t}{2}\right)\; \epsilon_t\\
& \varphi(B)\; (1-B)^{d} \;X_{t} = \vartheta(B) \;\eta_t,
\end{cases}
\end{equation}
where $\epsilon_t$ and $\eta_t$ are two independent iid sequences of Gaussian random variables and $\sigma(\cdot)$ is a known function.

\subsubsection{SISR for Fractional ARIMA process with known parameter}

We apply our algorithm to simulated data from an ARIMA(1, 0.3, 0) model with parameter $\varphi=0.8$.  That is
\begin{equation}\label{eq:arima}
\begin{cases}
&Y_{t} = |X_t|\; \epsilon_t\\
&(1-\varphi\; B)\; (1-B)^{d} \;X_{t} =\eta_t,
\end{cases}
\end{equation}
The simulated model is shown in  Figure \ref{figsim1}(a) and the estimated filter using the SISR algorithm is depicted in  Figure \ref{figsim1}(b). We choose as tuning parameters $\delta=0.98$ which then implies $(\alpha,h)=(0.9506, 0.096)$.

 \begin{figure}[!h]
 \centering
\begin{tabular}{cc}
 \includegraphics[scale=0.3]{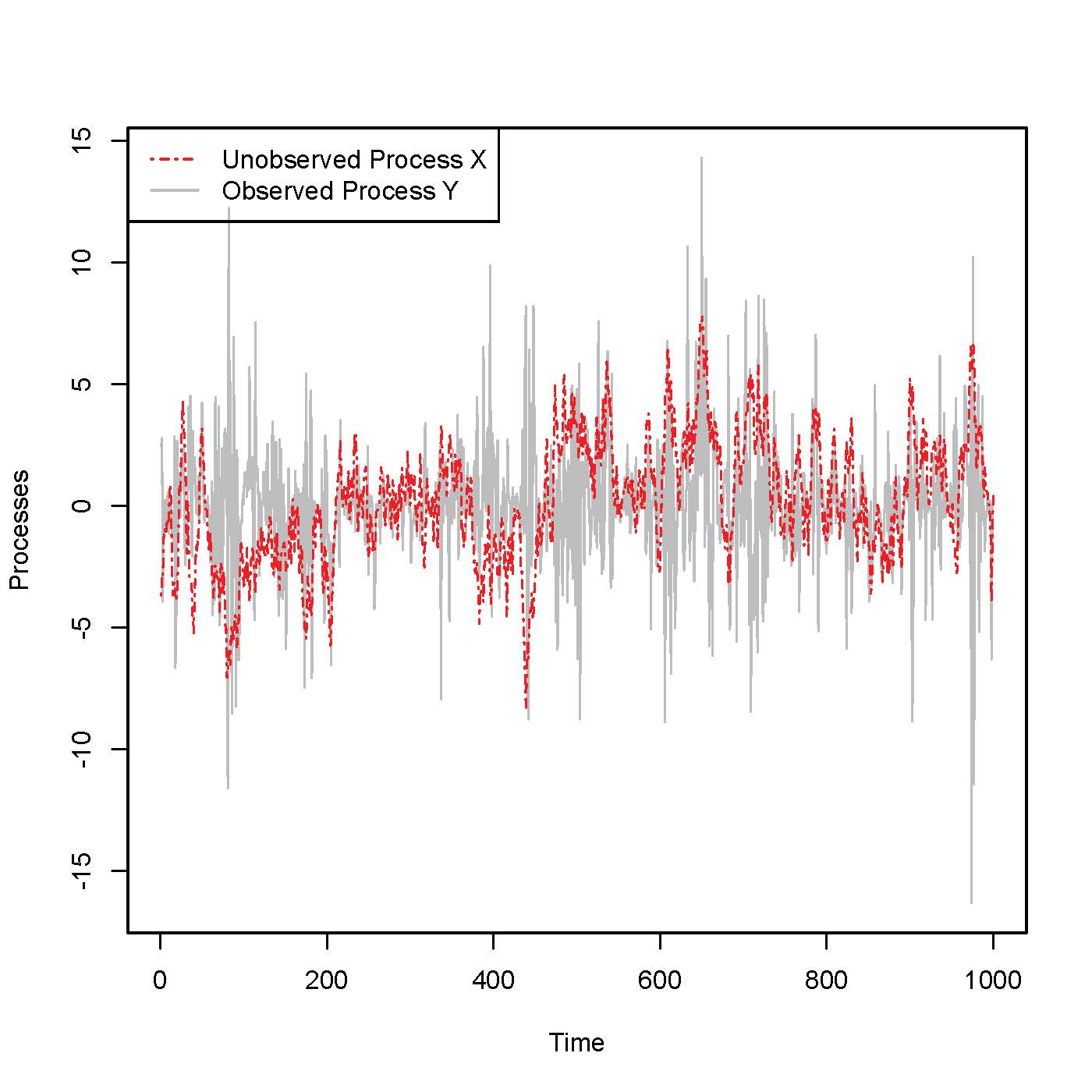}
 &
 \includegraphics[scale=0.3]{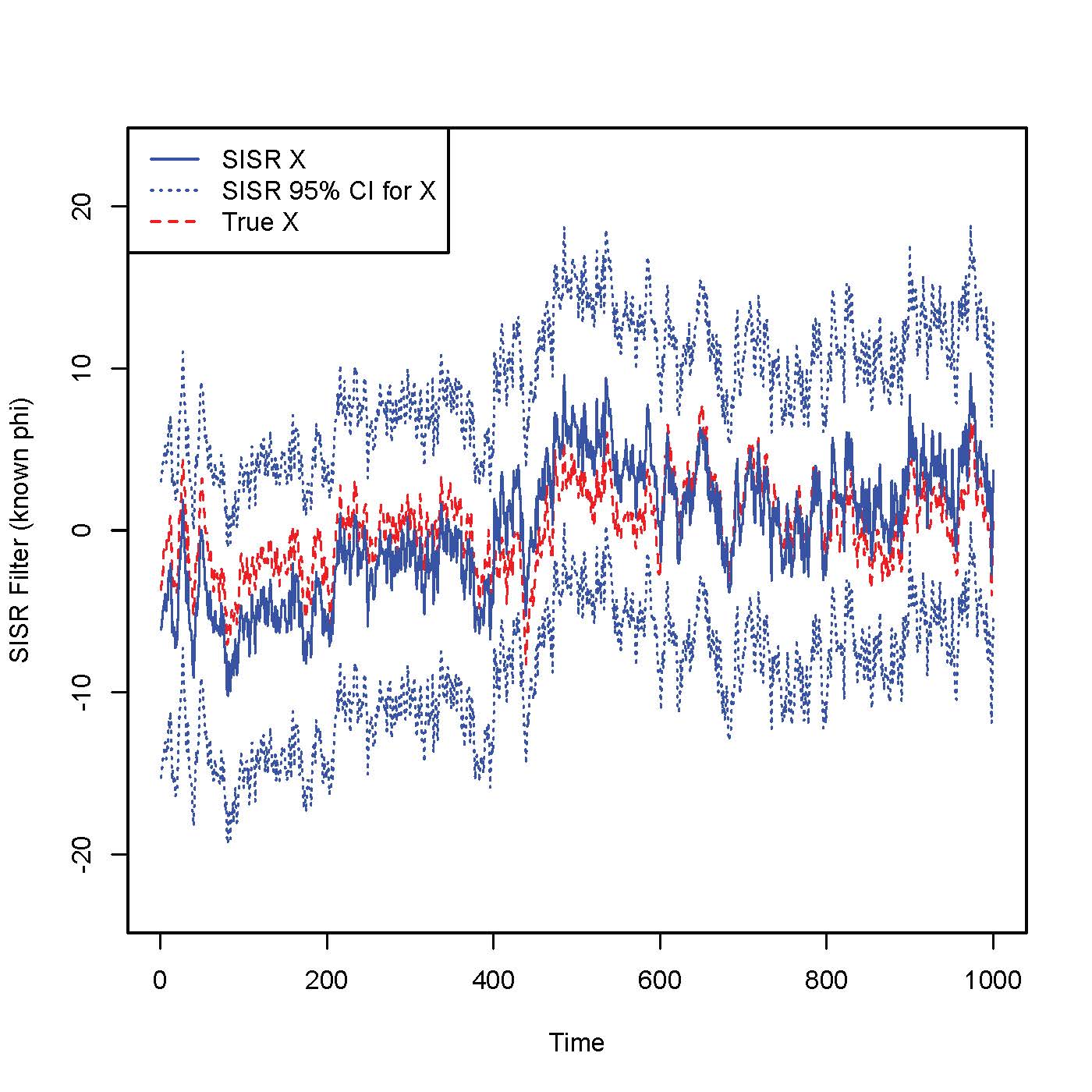}\\
 (a) Simulated Model &  (b) SISR Filter\\
 \end{tabular}
\caption{SISR filter for the fractional ARIMA(1, 0.3, 0) model.  Here, the  parameter $\varphi$  is assumed to be known, $\varphi=0.8$. } \label{figsim1}
 \end{figure}

From the SISR plot (Figure \ref{figsim1}b), we can see that the approximation of the unobserved process (solid  line) follows closely the true process (dashed line). In addition, the true data, all lie within the 95\% confidence interval (dotted lines) estimated using the SISR algorithm.
In our simulation study, we tried different values of $d$ and the results are similar.

\subsubsection{SISR for Fractional ARIMA process with unknown parameter}

Consider again model \eqref{eq:arima}, but now assume that the parameter $\varphi$ is unknown. Our goal is to estimate $\varphi$ using the SISR algorithm. The results are summarized in Figures \ref{figsim2} and \ref{figsim3}. From Figure \ref{figsim2}, we can see that the approximation of the unobserved process remains good, and the 95\% confidence interval still captures the process.

  \begin{figure}[!h]
  \centering
 \includegraphics[scale=0.3]{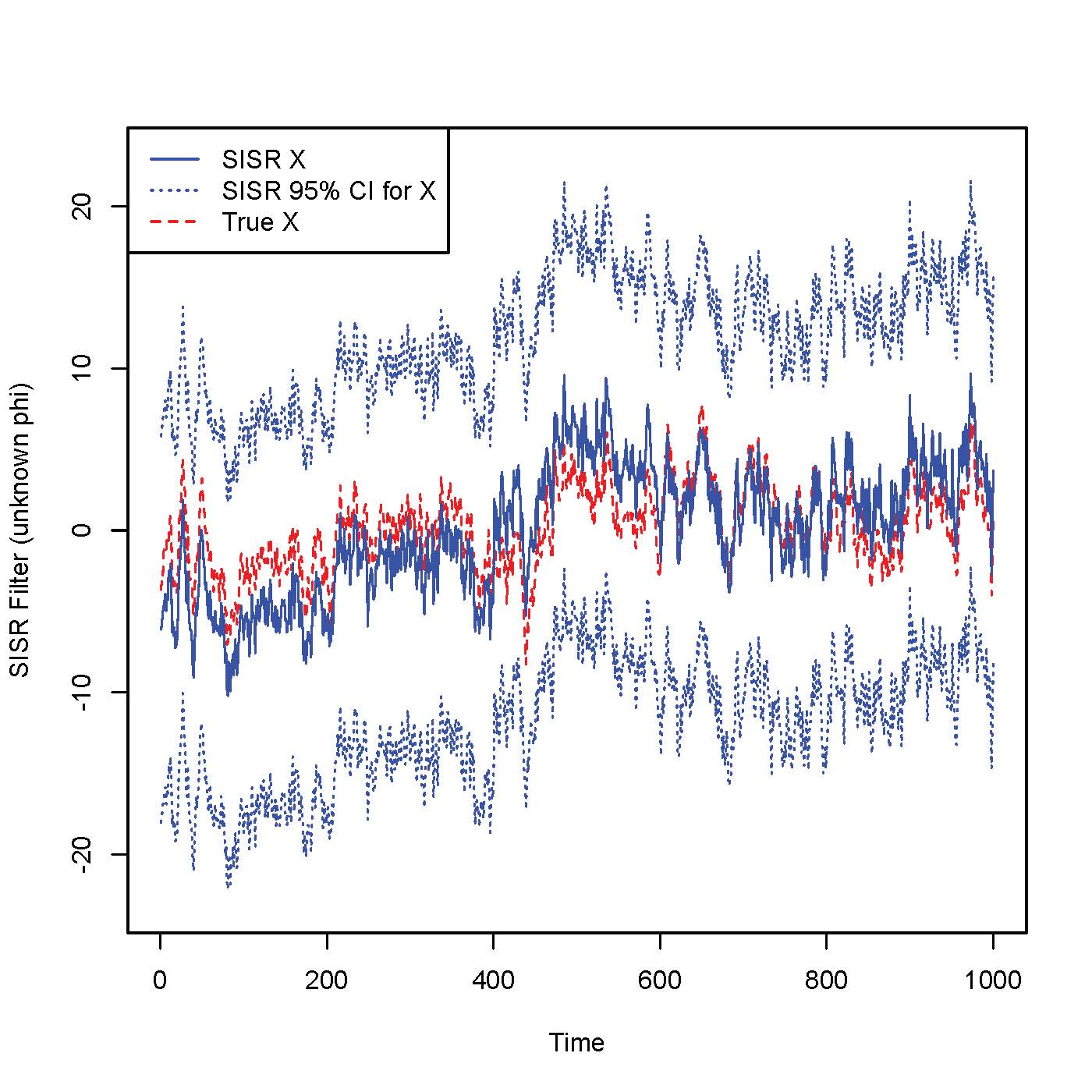}
\caption{SISR filter for the ARIMA(1, 0.3, 0) model. In this case, the  parameter $\varphi$  is assumed to be \textbf{unknown} and is estimated from the algorithm. }  \label{figsim2}
  \end{figure}

In Figure \ref{figsim3}, we investigate the convergence of the parameter to the true value. The estimated parameter $\varphi$ is slightly noisy, but it converges to the true value 0.8. The difference in the two graphs in Figure \ref{figsim3}, is that the second one has a significantly larger number of simulated particles, and it seems that this slightly improves the smoothness of the curve. 

In Figure \ref{figsim2}, we compare the empirical variance of the estimator for $N=500$ versus $N=2500$ across all $t\in[1,1000]$. It is clear that the variance decreases as N increases, but at the same time the variance does increase over time. The latter is consistent with the theoretical limiting variance as obtained in (\ref{Eq:VarianceEstimator}).

  \begin{figure}[!h]
  \centering
 \begin{tabular}{cc}
 \includegraphics[scale=0.3]{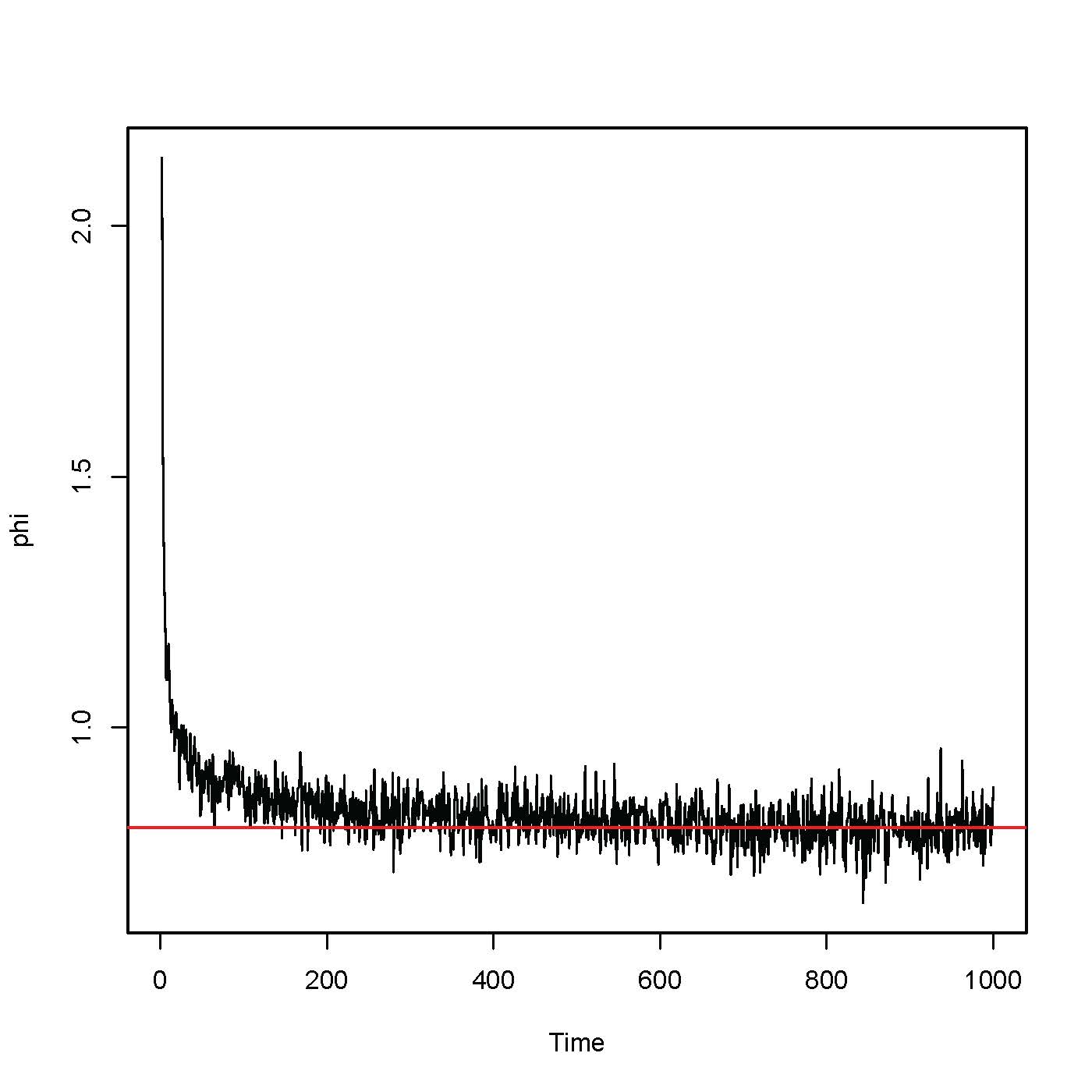}
 &
 \includegraphics[scale=0.3]{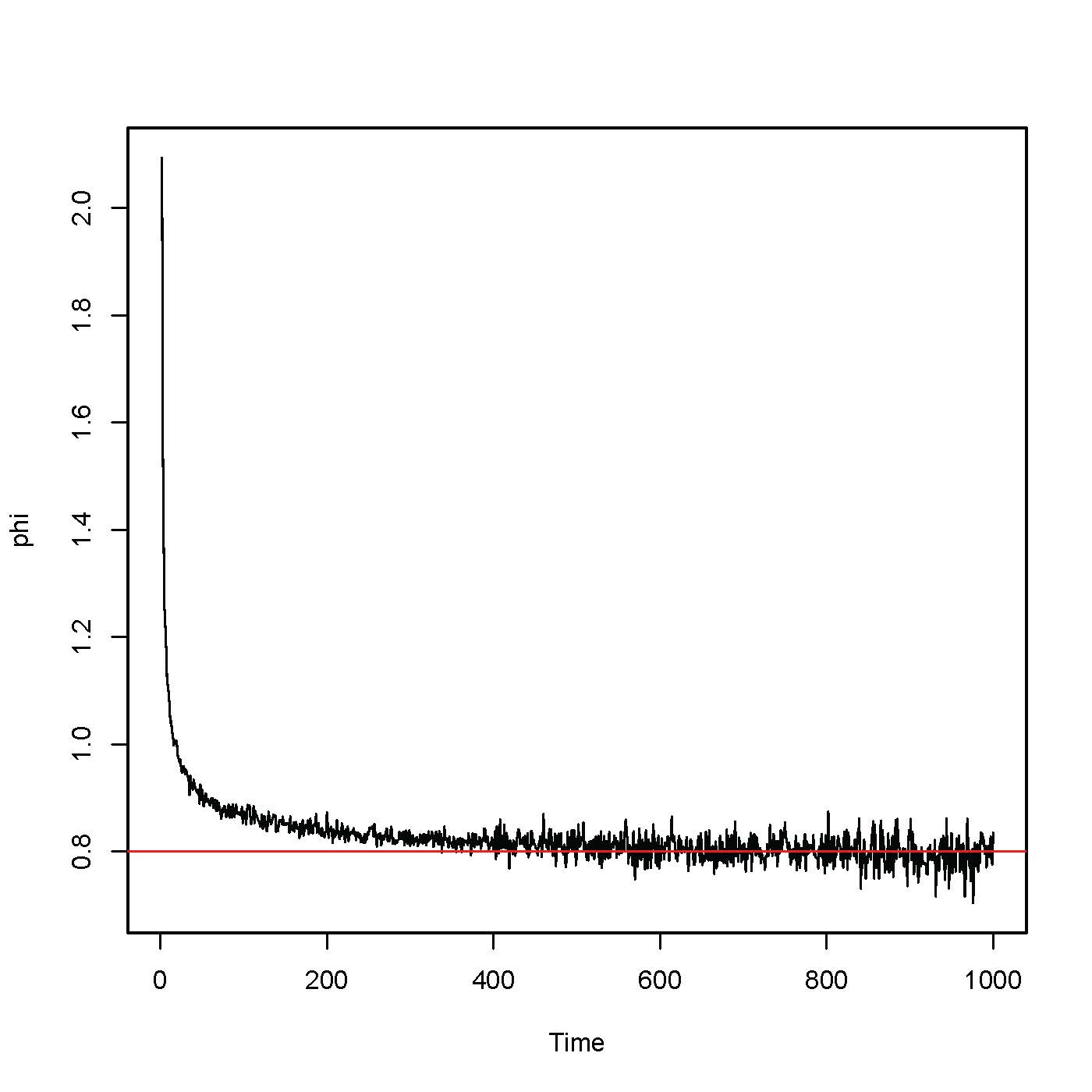}\\
 (a) Number of Particles $N$=500 &  (b) Number of Particles $N$=2,500
 \end{tabular}
\caption{Convergence of the estimated parameter as a function of time, for two different choices of number of particles.}  \label{figsim3}
  \end{figure}

  \begin{figure}[!h]
  \centering
 \begin{tabular}{cc}
 \includegraphics[scale=0.3]{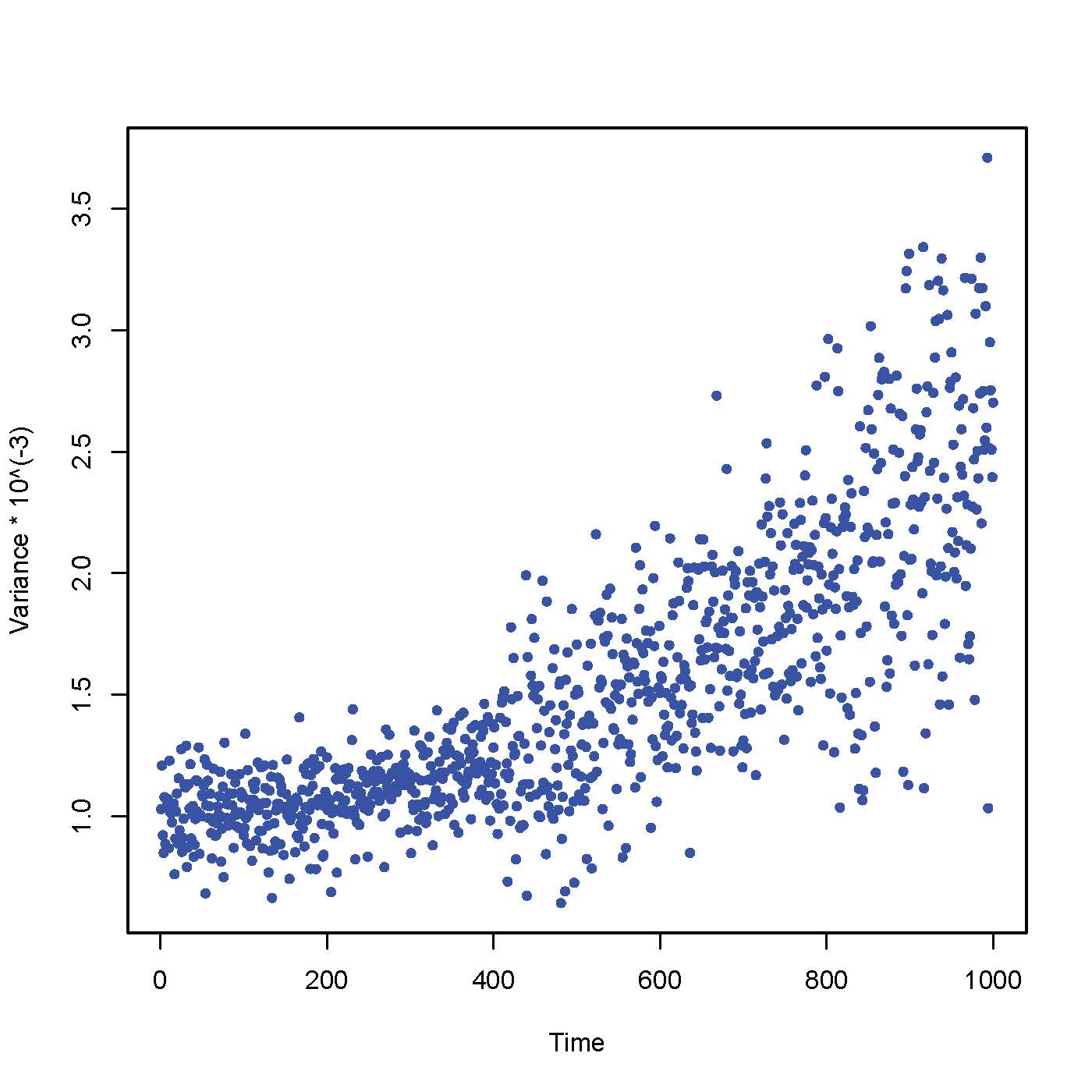}
 &
 \includegraphics[scale=0.3]{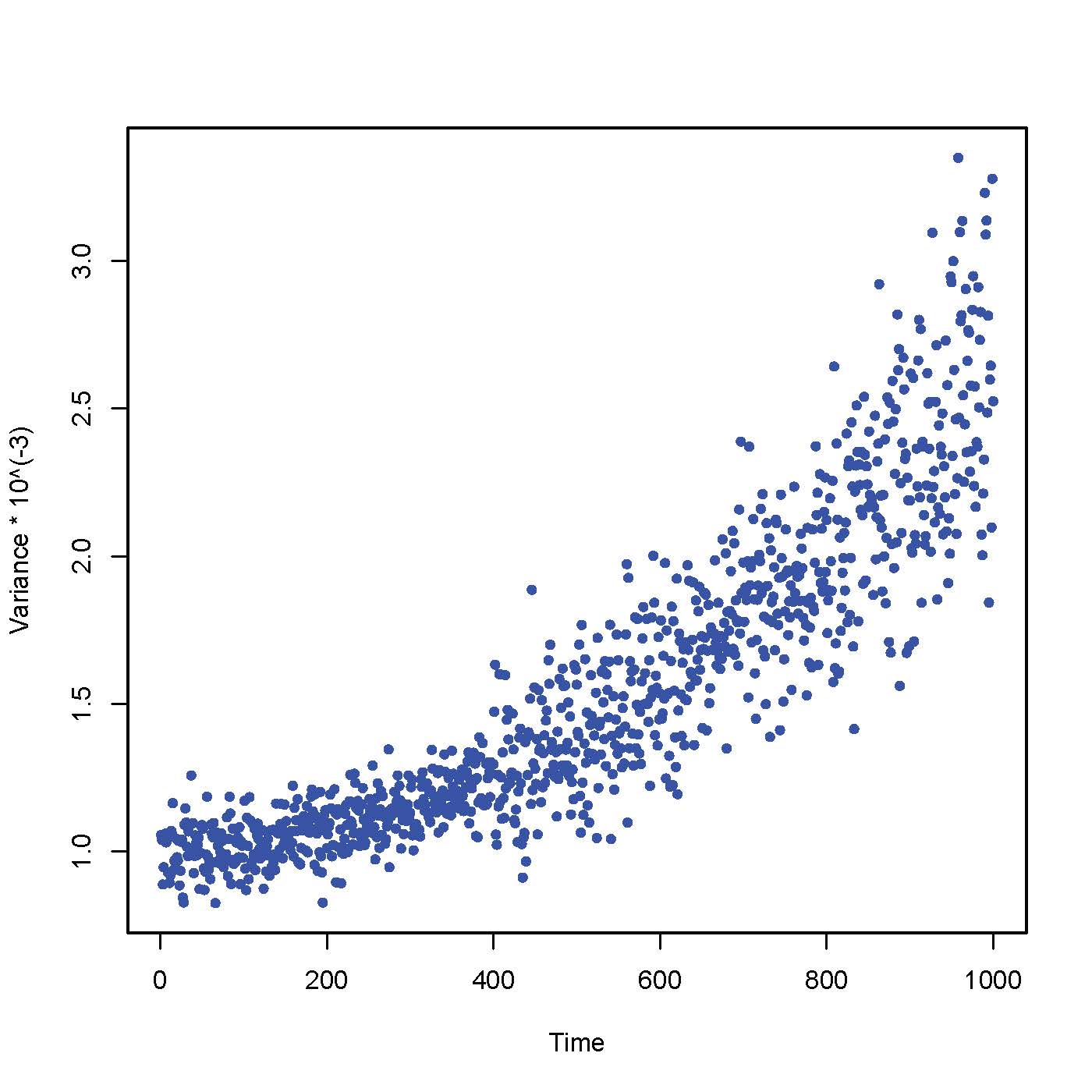}\\
 (a) Number of Particles $N$=500 &  (b) Number of Particles $N$=2,500
 \end{tabular}
\caption{Empirical Variance of the estimated parameters as a function of time for N=500 and N=2,500}  \label{figsim2}
  \end{figure}


\section{Application to S\&P 500 Data}

In this section, we apply our method to  real data. As an example, we are working with a long-range dependent state-space  model in finance. The observed process are the returns of the underlying asset (S\& P 500 index to be precise)  and the unobserved process is the asset's volatility. Based on the financial literature (\cite{BCL,Ha,chrono2}), we assume that the volatility process is long-range dependent, and the model we  focus on is the long memory stochastic volatility model in discrete time that is described by  \eqref{eq:lmsv}. This model was introduced simultaneously by  Breidt et al. \cite{BCL} and Harvey, \cite{Ha} in 1993.

 To further specify this model in practice, we need to determine the  order of the polynomials $\phi(\cdot)$ and $\theta(\cdot)$.  This is a  common task in time series analysis and for details we refer to Hamilton \cite{ham}. Based on a preliminary analysis, we choose to work  with a Fractional ARIMA($1,d,1$) model, which is also  in accordance to the model suggested by \cite{BBM}. To be precise, the model we will be working with is
\begin{equation*}
\begin{cases}
& Y_{t} = \sigma\left(\frac{X_t}{2}\right)\; \epsilon_t\\
& (1-\varphi B) \; (1-B)^{d} \;X_{t} = \vartheta \;\eta_{t-1} + \eta_t,
\end{cases}
\end{equation*}
where $\sigma(x) = \log x$.

The data set we consider contains daily returns of the S\&P 500 for one year, that is about 252 observations, starting in January 2010 until December 2010. 

One assumption that we made in the SMC algorithm  is that the parameter $d$  is known. However, when it comes to real data, this is something that we need to estimate. Here, we used the  Geweke and Porter-Hudak estimator, \cite{casas}, which yields $d=0.2$. Then, we apply the  SISR algorithm to estimate the remaining unknown parameters of the model  $\varphi$ and $\vartheta$. We choose as tuning parameters $\delta=0.986$ which then implies $(\alpha,h)=(0.965, 0.068)$.

The algorithm has two outputs. The first one is the distribution of the unobserved volatility, which is  given in Figure \ref{figsim5}, using 500 trajectories, and the second one is the estimated vector of parameters, which are plotted as a function of time in Figure \ref{figsim6}.

  \begin{figure}[!h]
  \centering
 \includegraphics[scale=0.3]{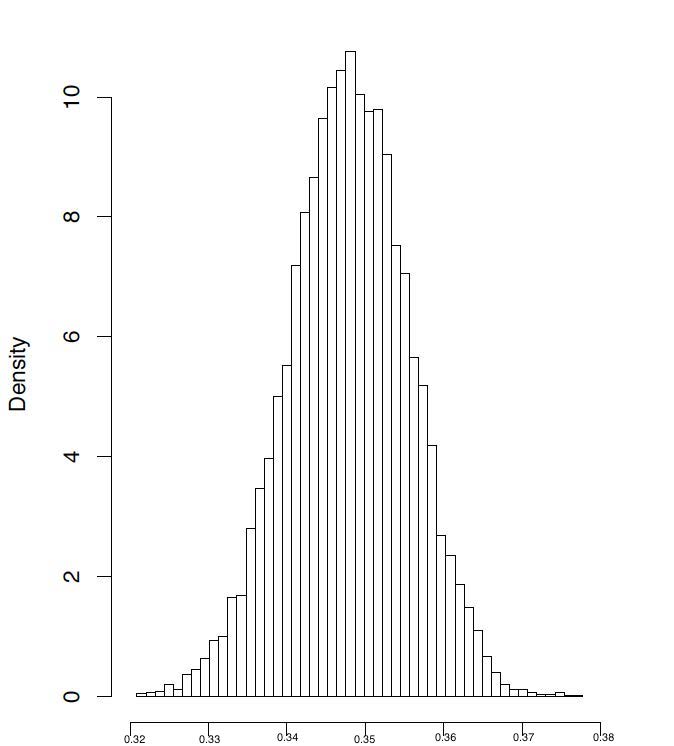}
\caption{Histogram of the empirical distribution of the unobserved volatility. As a reference, the implied volatility for the same period was 0.355.}  \label{figsim5}
  \end{figure}

  \begin{figure}[!h]]
  \centering
  \begin{tabular}{cc}
 \includegraphics[scale=0.3]{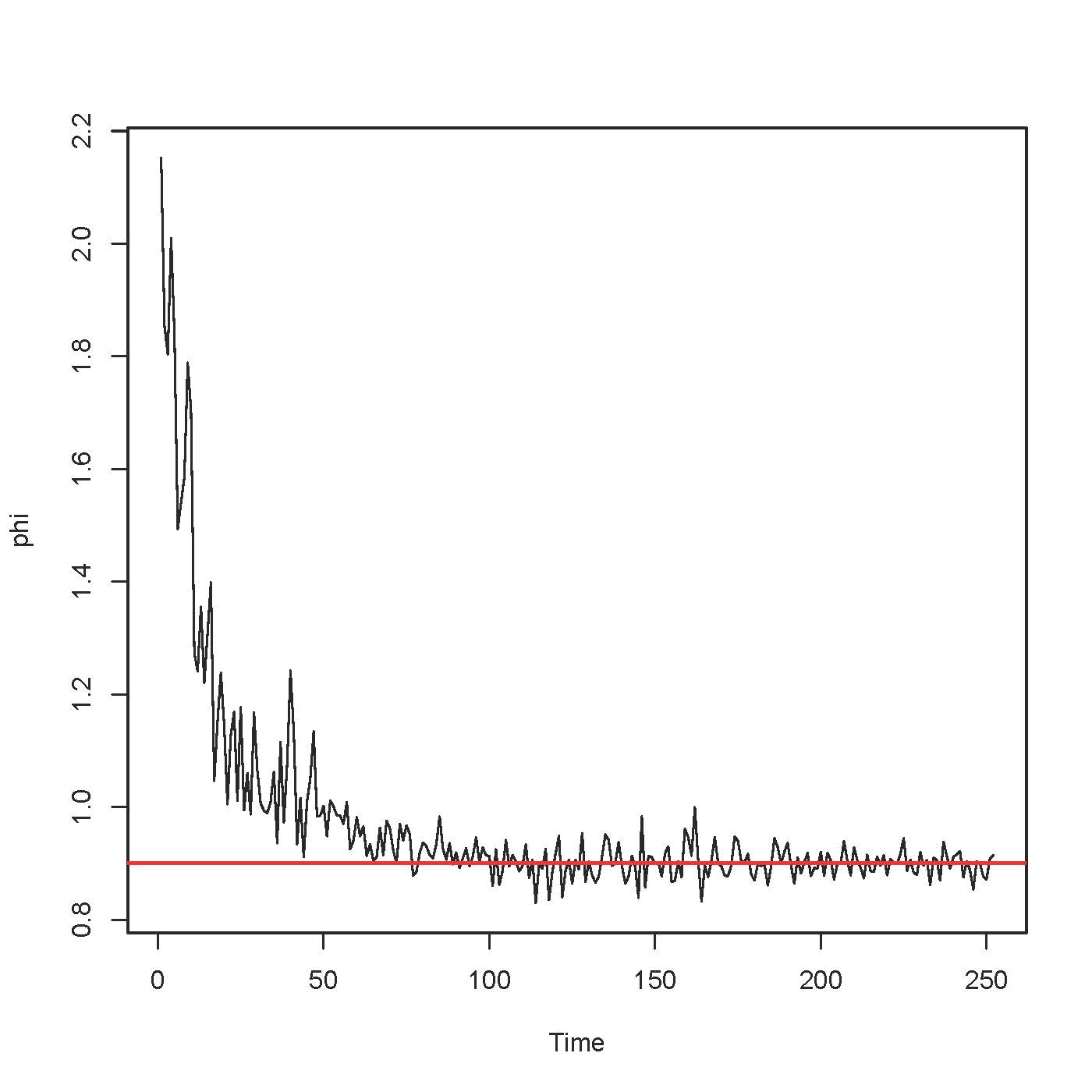} &
 \includegraphics[scale=0.3]{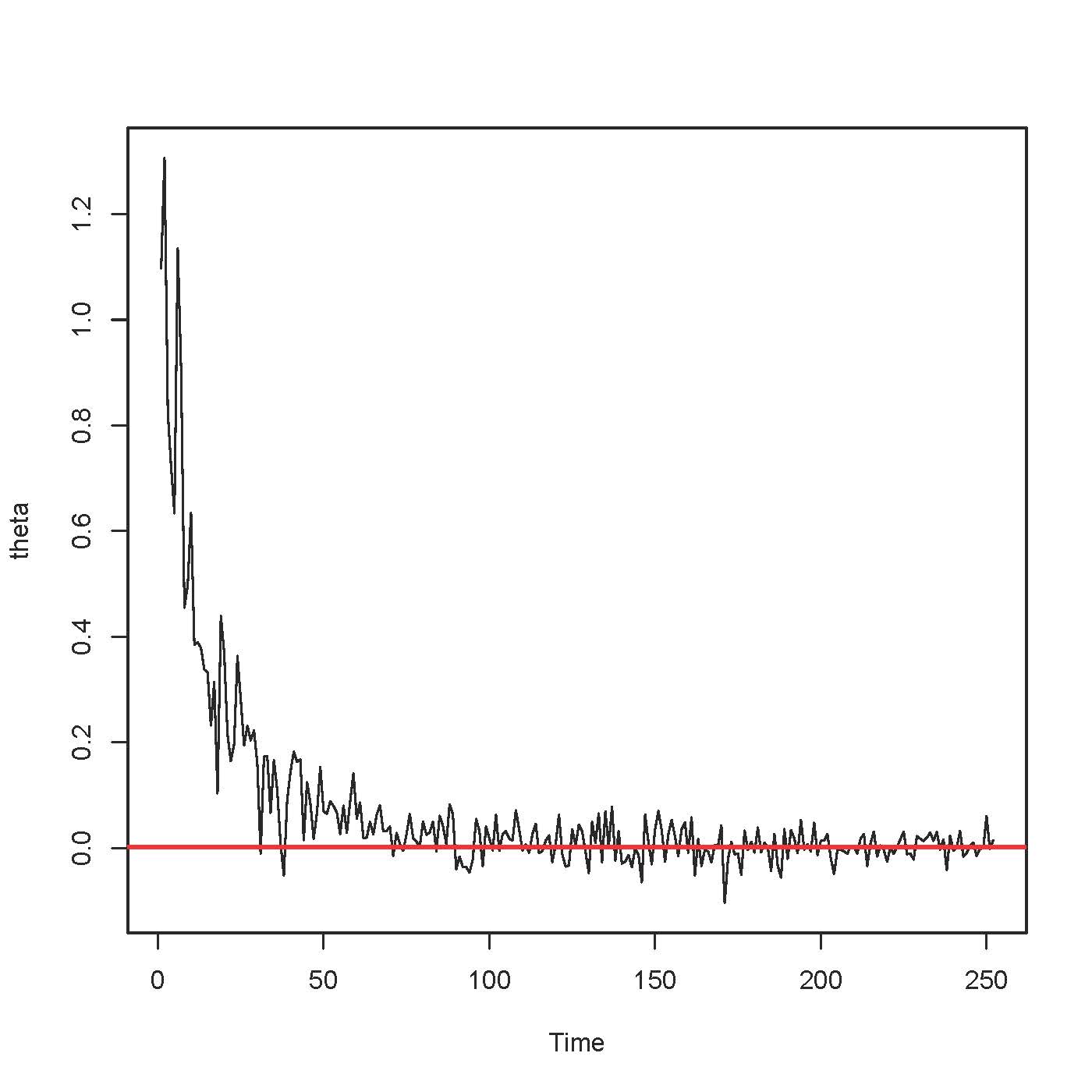}\\
 (a) Estimator of $\varphi$ & (b) Estimator of $\vartheta$\\
  \includegraphics[scale=0.3]{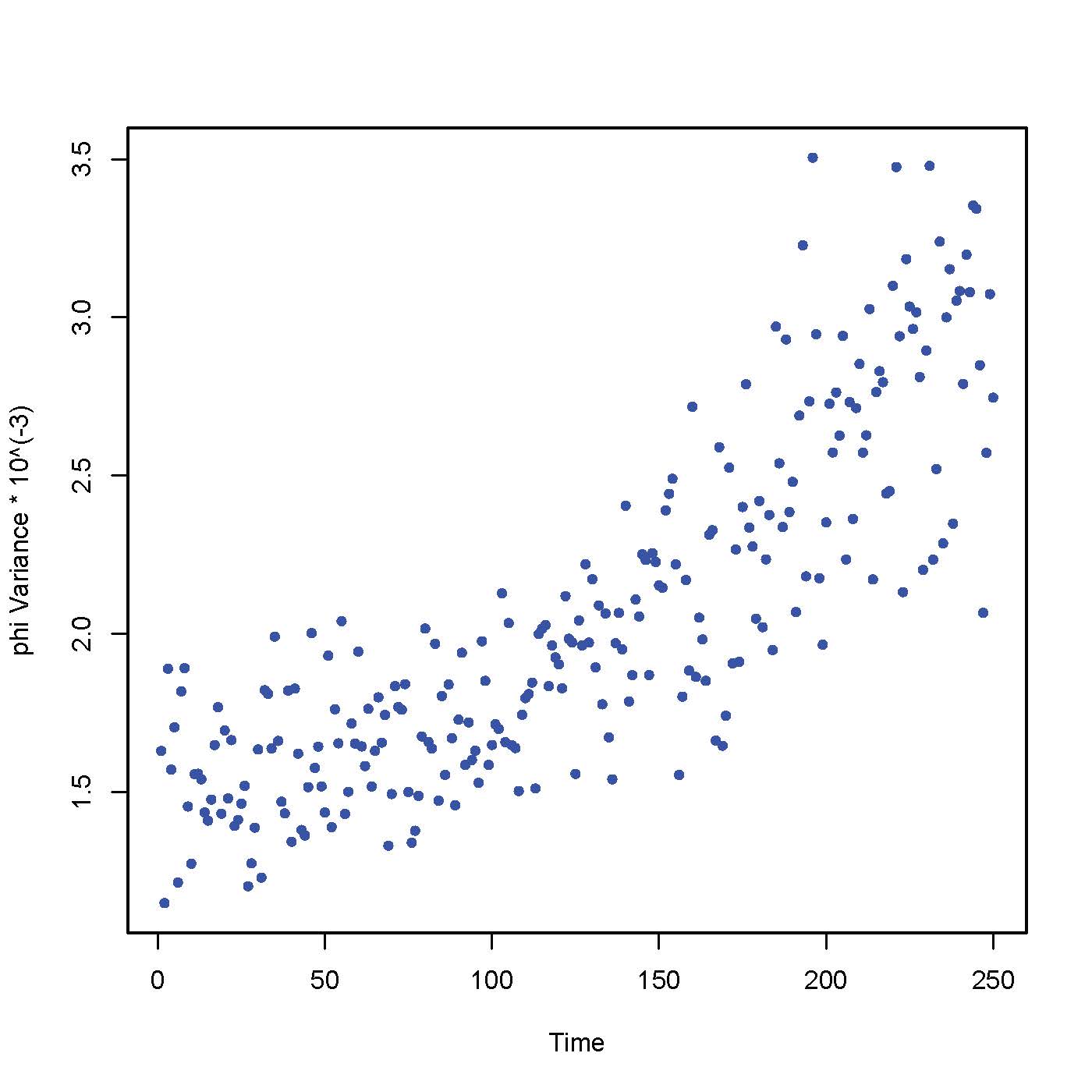} &
 \includegraphics[scale=0.3]{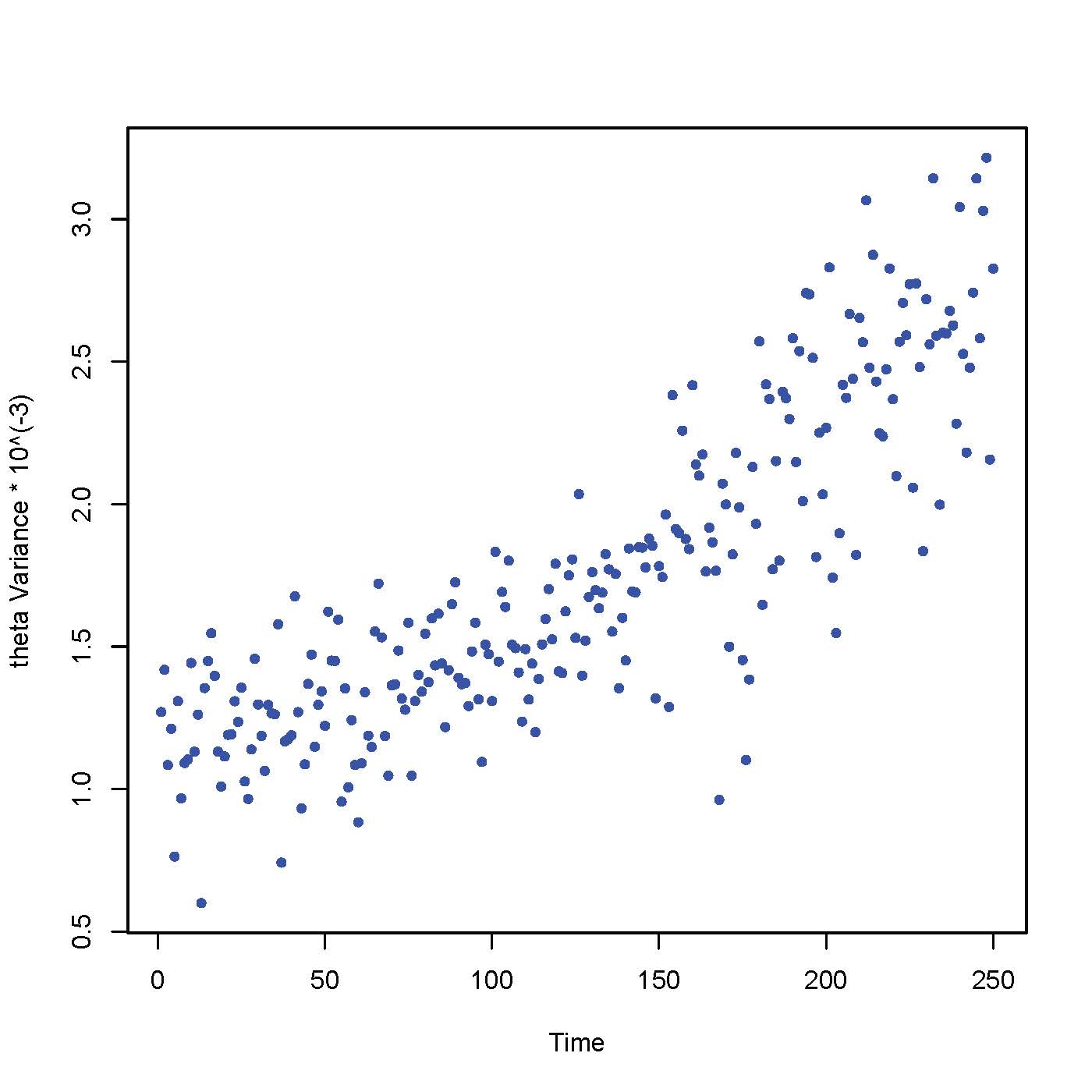}\\
 (c) Empirical Variance of $\hat{\varphi}$ & (d) Empirical Variance of $\hat{\vartheta}$
 \end{tabular}
\caption{Estimators for $\varphi$ and $\vartheta$ for the S\&P 500 data set. Solid lines represent the estimated values for $\hat{\varphi}=0.842$ and $\hat{\theta}=0.01$ for $\phi$ and $\theta$ respectively. }  \label{figsim6}
  \end{figure}

Using the model parameters we estimated, we also do an out-of-sample prediction of the values of the underlying asset, which is shown in Figure \ref{figpred}. By doing an one-step prediction each time, we forecast 20 daily values of the index. The 95\% confidence intervals are computed using boostrap. In Figure \ref{figpred} we also present the empirical variance for the estimators $\hat{\varphi}$ and $\hat{\theta}$. In Figure \ref{figresids}, we also present the residuals of the fitted model.

  \begin{figure}[!h]
  \centering
 \includegraphics[scale=0.3]{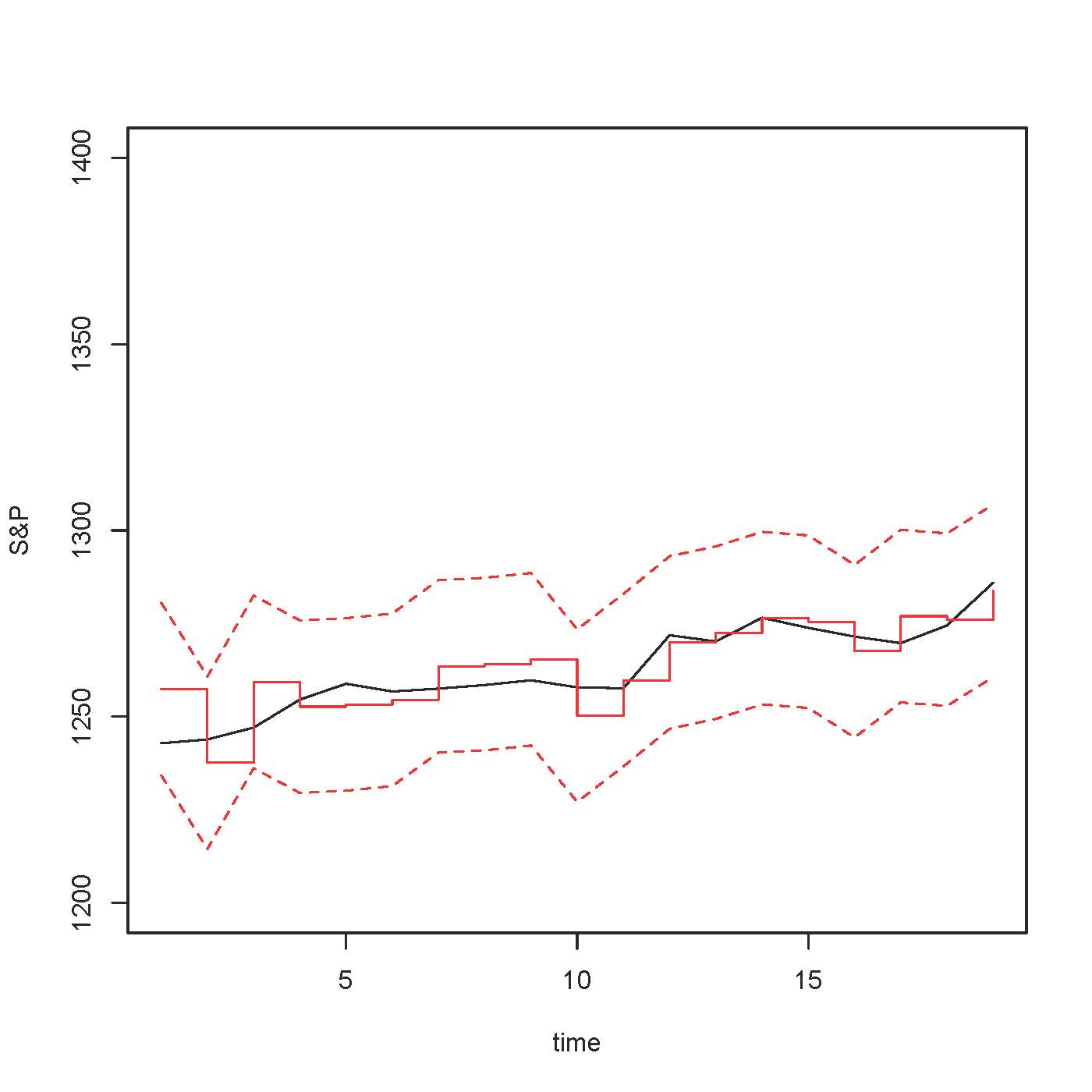}
\caption{Out-of-Sample Prediction of the S\& P 500 values. The (black) solid continuous line are the true S\& P 500 values, while the (red) solid step  line is our estimation. The (red) dotted lines form the 95\% prediction interval.}\label{figpred}
  \end{figure}

  \begin{figure}[!h]]
  \centering
  \begin{tabular}{cc}
 \includegraphics[scale=0.3]{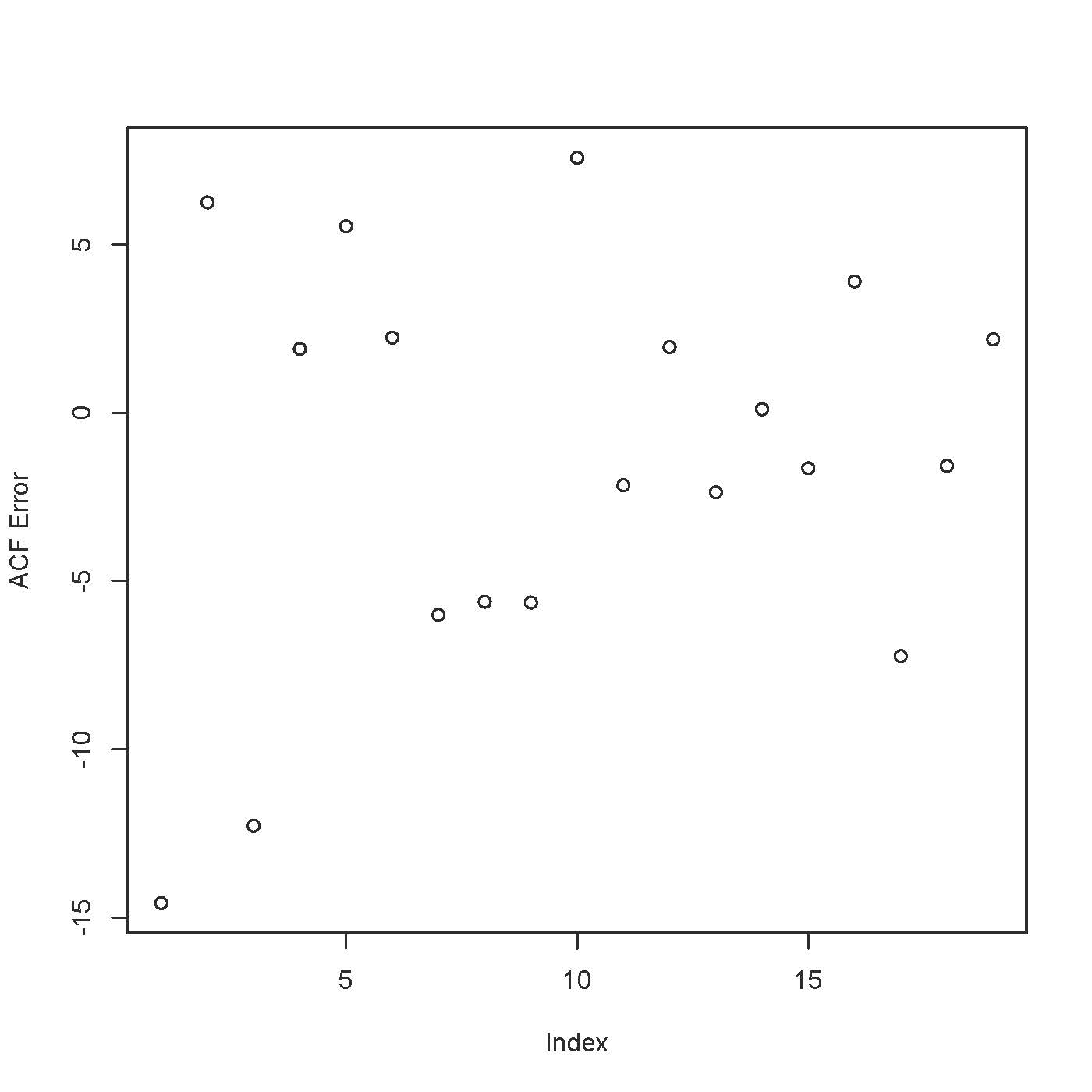} &
 \includegraphics[scale=0.3]{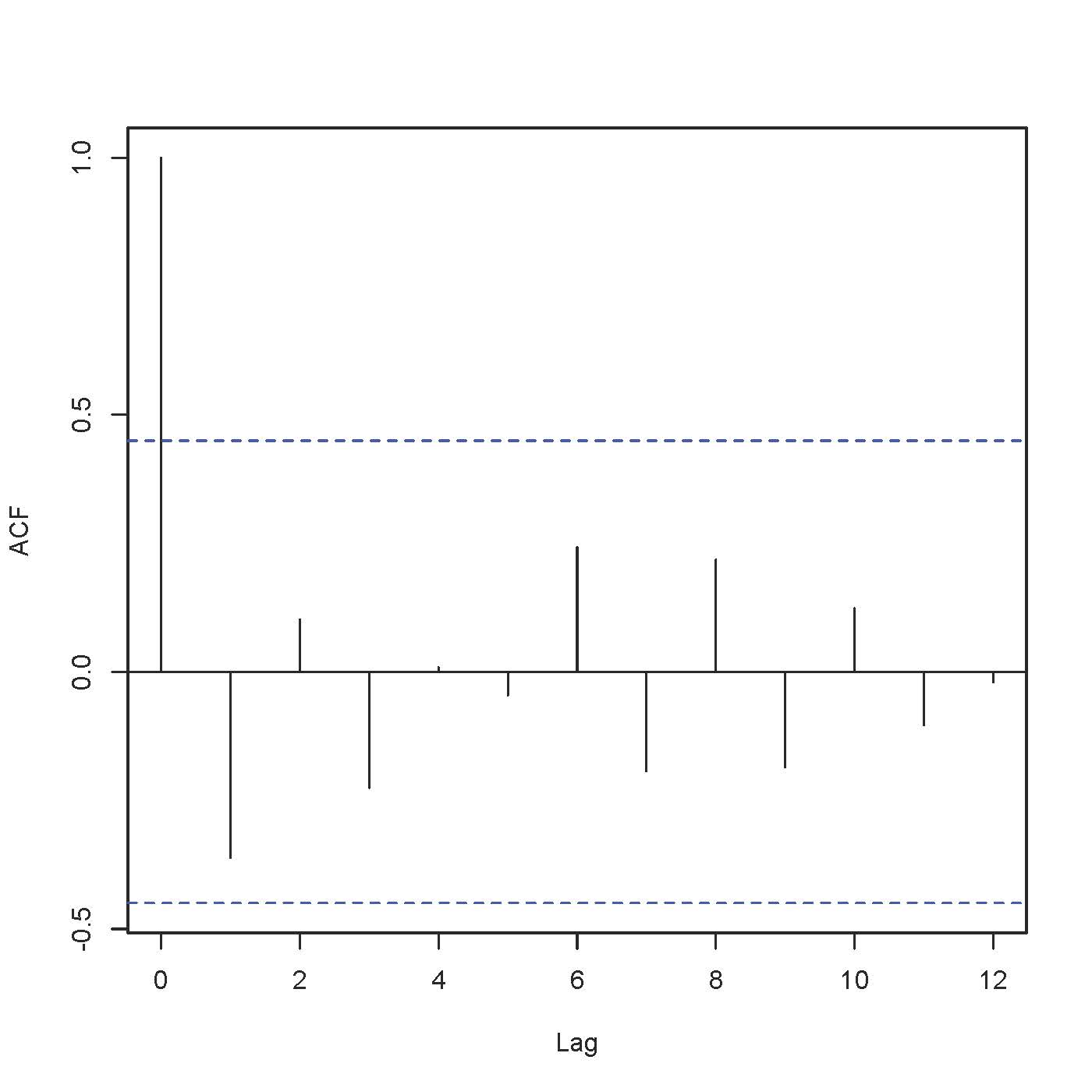}\\
 (a)  Standard Residuals  & (b) ACF of the Residuals
 \end{tabular}
\caption{Residuals of the fitted fractional ARIMA (1, 0.2, 1) model with $\varphi=0.842$ and $\vartheta=0.01$.}  \label{figresids}
  \end{figure}

\section{Conclusion}

To summarize, in this article we extended the standard SISR algorithm to incorporate the case that the observations are long-range dependent. Our findings show that the results are very close to the case that the observations are independent or Markov. However, the main drawback of this method is the computational time that is required to perform the iterations. Since we need to take into account, and technically speaking to store all past values of the trajectory, this increases the computational time and complexity of the method. In addition, by naturally extending existing results in the literature, we proved that the filter converges to the true distribution of the unobserved process.

Our second outcome, was the development of an SISR algorithm that along with the estimation of the unobserved distribution of the hidden process, also estimated  unknown model parameters. Our approach was dynamic, in the sense that the parameter was regarded as ``time-varying'' and thus the parameter estimators were updated at every step of the algorithm. We also showed that the proposed estimators for the unknown parameter are consistent and asymptotically normal and we corroborated these results with a simulation study.

There are quite a few open problems that we would like to investigate in the future. The first one is to study ways to improve the computational efficiency of the algorithm. In our approach, we used all the history of the trajectory to run the algorithm, which severely affected the computational efficiency, but it would be interesting to investigate if a ``window" approach would provide us with a reasonable estimator for the filter and/or the parameter, and possibly quantify the loss that one might have by doing so in terms of accuracy.

In addition, one question that we did not address in this paper, is what happens with the long memory parameter in practice. In our approach, we assumed that $d$ (or equivalently $H$) is known (given or estimated from the data). However, it is an open question how one would consistently estimate the memory parameter in the scenario that the long-range dependent process is hidden.

The goal of this first paper on the topic is to lay down the algorithm, its properties and to understand the main issues that the presence of memory brings. In this paper, we formulated the algorithm and established some baseline theoretical properties. We also performed numerical experiments using both simulation data and real data, in order to test the performance of the proposed algorithm in practice. We plan to investigate the open computational and theoretical issues mentioned above in a systematic way in future works.

\end{document}